\documentclass[twocolumn,floatfix,showpacs,amsmath,amssymb,pra]{revtex4}
\usepackage{times}
\usepackage{graphicx}
\usepackage{float}
\usepackage[caption = false]{subfig}

\usepackage{dcolumn}
\usepackage{bm}
\usepackage{mathtools}

\usepackage{xparse}
\NewDocumentCommand{\ceil}{s O{} m}{%
	\IfBooleanTF{#1} 
	{\left\lceil#3\right\rceil} 
	{#2\lceil#3#2\rceil} 
}

\newcommand\disteq{\stackrel{\mathclap{\normalfont\mbox{\scriptsize d}}}{=}}

\draft

\usepackage{amsfonts, amsmath, amssymb, amsthm, graphicx,mathtools }
\usepackage{natbib}
\newtheorem{theorem}{Theorem}
\newtheorem*{theorem*}{Theorem}

\newtheorem{definition}[theorem]{Definition}
\newtheorem{example}[theorem]{Example}

\usepackage{bbold}

\begin{document}
\title{(Quantum) Fractional Brownian Motion and Multifractal Processes under the Loop of a Tensor Networks} 
\author{Beno\^it Descamps}%
\email{benoit.descamps@univie.ac.at}
\affiliation{Faculty of Physics, University of Vienna, Austria}%
\affiliation{Department of Physics and Astronomy, University of Ghent}%
\date{\today}

\begin{abstract}
We derive fractional Brownian motion and stochastic processes with multifractal properties using a framework of network of Gaussian conditional probabilities.
This leads to the derivation of new representations of fractional Brownian motion. These  constructions are inspired from renormalization.
The main result of this paper consists of constructing each increment of the process from two-dimensional gaussian noise inside the light-cone of each seperate increment. Not only does this allows us to derive fractional Brownian motion, we can introduce extensions with multifractal flavour. 
In another part of this paper, we discuss the use of the  multi-scale entanglement renormalization ansatz (MERA), introduced in the study critical systems in quantum spin lattices, as a method for sampling integrals with respect to such multifractal processes. After proper calibration, a MERA promises the generation of a sample of size $N$ of a multifractal process in the order of $O(N\log(N))$, an improvement over the known methods, such as the Cholesky decomposition and the circulant methods,  which scale between $O(N^2)$ and $O(N^3)$.

\end{abstract}

\maketitle
\section{Introduction}
The study of long memory processes \cite{Beran_1994_statistics} is an old story. When sampling the same population at different point in  times, $X_1,..., X_n$, we often hope and assume that the time-average reflects the local average,
$$\overline{X}= \frac{1}{n}\sum_j X_j,~~E(\overline{X}) \approx \frac{1}{n}E(X_j)$$
The central limit theorem tells us that this is indeed true if $X_j$ are identically distributed and are independent. The variance is then inversely proportional with time, the sample size, and thus decay.
This fact remains true even when $X_j$ have weak correlations such as an exponential decay \cite{attal2012central, Sinayskiy_2013_OQW}.
This universality of the inverse proportionality with the sample size, breaks down when the correlations are stronger and decay polynomially. This phenomena of long range dependence is very well studied in many areas of physics, from statistical physics to quantum field theory.

In statistics, many classes of gaussian processes with long-range depency have been studied throughout history. One of the most studied and well known is certainly fractional Brownian motion.
These processes were first introduced by Kolmogorov in 1940 when investigating turbulences. Perhaps it is Benoit Mandelbrot who truly recognized the importance of the fractional Brownian motion. In 1965, he published his insights on the work of the hydrologist Harold Erwin Hurst, who observed discrepancies in the yearly variation of the levels of the Nile river \cite{Koutsoyiannis_2002_hurst,Mandelbrot_1968_fractional,Mandelbrot_1999_turbulence}.
\begin{figure}[t!]
	\includegraphics[width=0.4\textwidth]{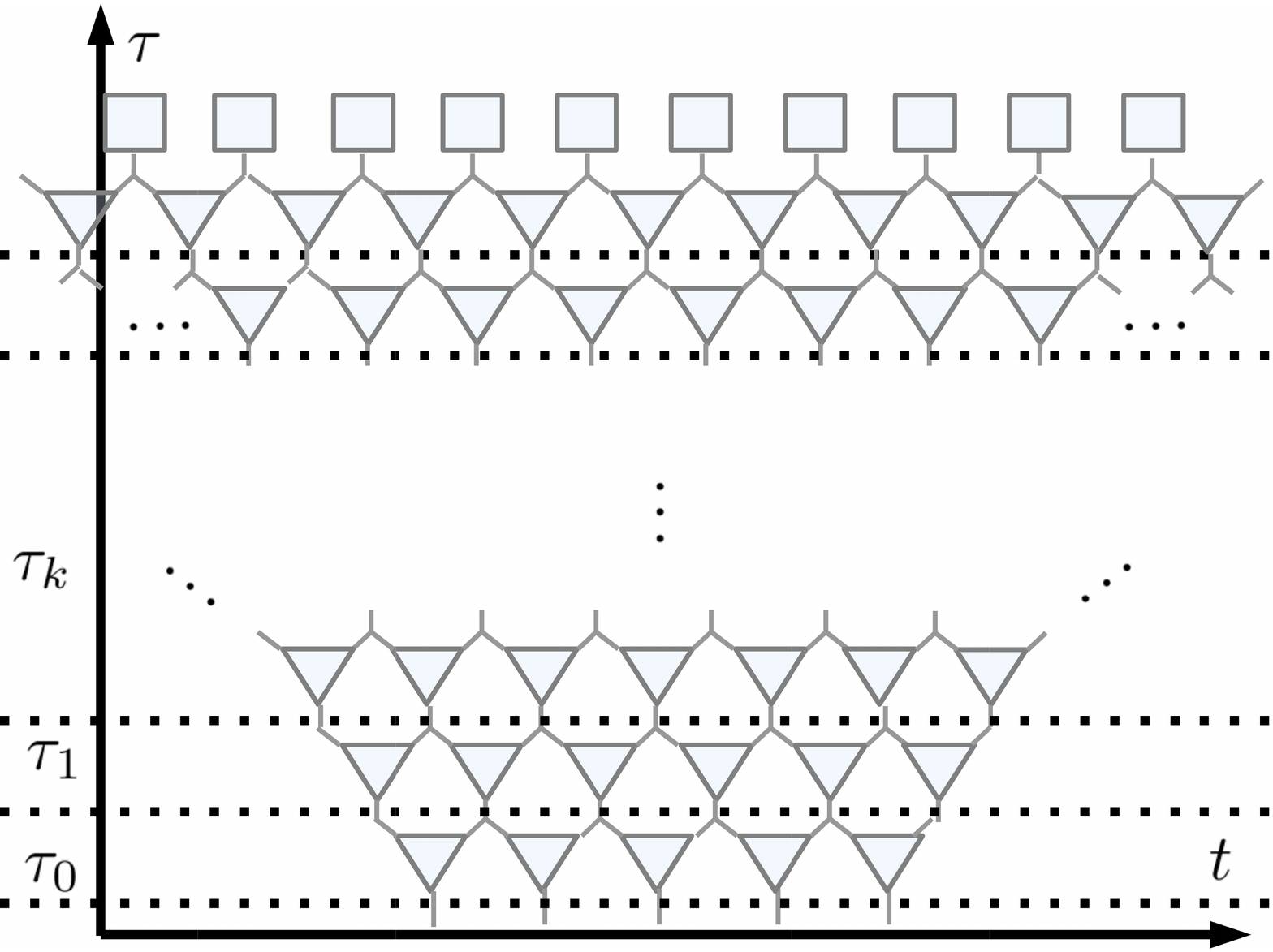}
	\caption{Network representation of the joint probability distribution of a stochastic process.  }
	\label{fig:CIRCUIT}
\end{figure}
Using a fractional integration of the Brownian motion,  the following process $B_H(t)$ is introduced,
\begin{align}
B_H(t) = &\int_{-\infty}^{0} \{(t-s)^{H-1/2}-(-s)^{H-1/2}\}dB(s)\nonumber\\
&+\int_{0}^{t}(t-s)^{H-1/2}dB(s) \label{eq:int_fbm}
\end{align}
The constant $H \in (0,1)$ is also known as the Hurst index. For $H= 1/2$ this reduces to regular Brownian motion.

With the rise of computational power, new methods were developped for simulating condensed matter. The challenge faced in such system was the exponentially growing number of parameters. Solutions were presented in the form of various ansatz states \cite{Perez_2006_MPS,verstraete2008matrix}. These states were constructed from networks of tensors with a particular geometry.

In 2005, Vidal presented an tensor network ansatz reminiscent of renormalization to simulate quantum critical systems \cite{Vidal_mera}. While being mostly used for numerics, such scheme sparks various interests in other areas such as high energy physics \cite{nozaki2012holographic} etc... A continuum version was presented in 2010 by Haegemann et. al. \cite{Haegeman_cmera}.

In this work, we show that fractional brownian motion can be related with such networks (\ref{fig:CIRCUIT},\ref{fig:MERA}).  We start from a discrete process $(X_n)_{n=1}^N$ with a joint probability distribution $p(X_1,...,X_N)$ represented by a certain network. Under renormalization of the parameters of the network, we prove that the process $\sum_n X_n$ is precisely a fractional Brownian motion. The networks  also yield a new representation of fractional Brownian motion.

In the first part, we construct a new network with an underlying causal structure which follows from a renormalization flow. We will show that this network generates fractional brownian motion. The starting point of this work is the relationship between fractional brownian motion and renormalization \cite{hochberg2002renormalization}. Based on this knowledge, we also discuss the use of MERA for simulating such processes. While a MERA is more challenging to derive analitically, we can calibrate the parameters numerically. It turns out that the MERA structure offers the possibility of simulating such gaussian processes more efficiently, $O(N\log N)$, than the regular methods such as the Choleksy decompositon or circulant methods \cite{Dieker_simulation_fbm}.

The two networks (\ref{fig:MERA},\ref{fig:CIRCUIT}) presented in this work differ in the direction of the renormalization flow. While for MERA the real space is renormalized in the virtual space, this is quite the opposite in the second network.

\section{Network Representation of Stochastic Processes}

By compounding or integrating familar processes, Staticians derive new ones with new desired properties. The use of processes with simple distributions, such as Gaussian distributions, also permits the efficient simulations of the processes without having the derive the distribution of the new process. Sometimes, however, the opposite direction is necessary. This allows for a change of measure, which can further simplify the process. The most famous example is the so-called Girsanov theorem, which allows to eliminate the drift from a Brownian motion by change of measure \cite{girsanov1960transforming}. 

Time series \cite{Beran_1994_statistics} such as ARIMA, ARFIMA, etc... , are succesful techniques for tackling memory. These processes make use of the increments of a one-dimensional Brownian motion up to some time $t$,
$$ X_t = \int_0^t c(s)dB_s$$
At each later time, new increments are added.
In the era of tensor networks, besides the real, here time-, axis additional 	virtual dimensions are added. These new axes potentially and seemingly add new parameters, but present us with new insights in such processes.
The goals of this section is to discuss the expansion of the joint probability distribution $P(X_1,...,X_N)$ of the increments $\vec{X}_j$ of a process $Y_N = \sum_j X_j$ as a circuit of conditional probabilities. The circuit is set up along a virtual dimension, which we denote as $\tau$. The final output of the circuit is at $\tau_0$, which is then the seeked probability.
\begin{align*}
&P(\vec{X}_{t_j,\tau_0})=\sum_{j_1,...,j_{N_{\infty}}}P(\vec{X}_{t_{j_0},\tau_0}|\vec{X}_{t_{j_1}},\tau_1)\\&P(\vec{X}_{t_{j_1}},\tau_1|\vec{X}_{t_{j_2}},\tau_2)...P(\vec{X}_{t_{j_{N_\infty-1}}},\tau_{N_{\infty-1}}|\vec{X}_{t_{j_{N_\infty}}},\tau_{N_{\infty}})
\end{align*}
Both networks are intrinsically connected with renormalization. The first circuit (\ref{fig:CIRCUIT}) has a natural causal structure and is discussed in the next section.
We will also display the power of an already known network, MERA, shown in figure (\ref{fig:MERA}). Our  approach  in section (\ref{section:Numerical}) is then purely numerical.
\subsection{A Light-Cone network for fractional Brownian Motion}

In 1966, Leo P. Kadanoff proposed the "block-spin" renormalization group in his study of phase transitions of the Ising model. He hypothesized that because spins would line up in large blocks near the critical points, then neighbouring spins can be regrouped and treated as a single entity. This ansatz allowed him to rederive scaling laws near the critical point.

One of the properties of Fractional Brownian motion is the self-similarity of the process,
$$B^{\mbox{\scriptsize fm}}(at) \disteq |a|^{H}B^{\mbox{\scriptsize fm}}(t)$$
where the equality is understood as in distribution.
 This scaling property fits perfectly with the philosophy of Kadanoff's renormalization on a tree. If some parameter $m_{\tau_k}^{\tau}$ describes some group of tensors from a virtual time $\tau_k$ to $\tau$, then according to Kadanoff it should be a related by a rescaling when the block is made larger,
 $$m_{\tau_k}^{a\tau} \propto a^{H'} m_{\tau_k}^{\tau}$$
 \begin{figure}[t!]
 	\includegraphics[width=0.4\textwidth]{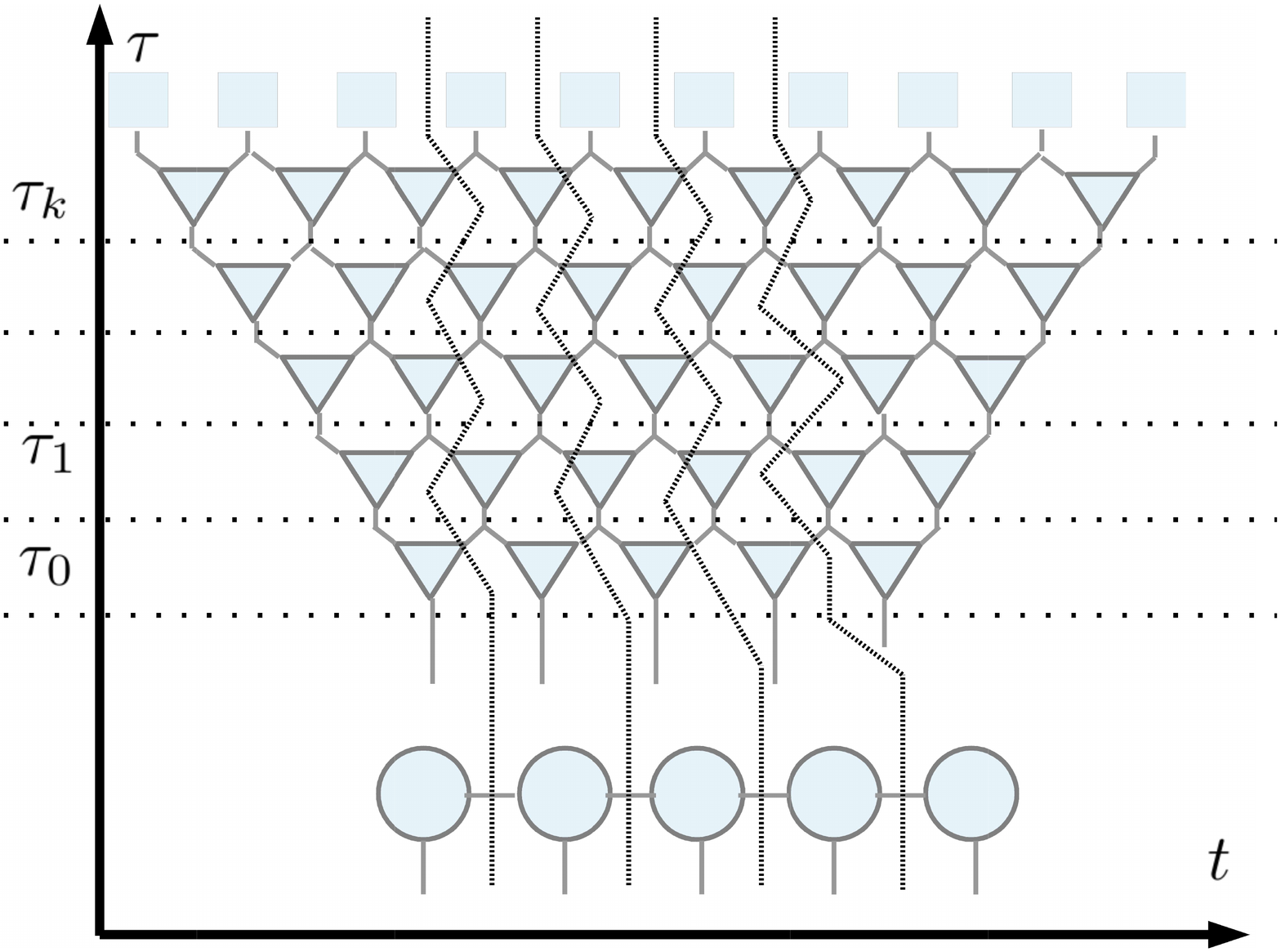}
 	\caption{The joint probability distribution can be approximated by a finite debt circuit. This in turn can be represented by a so-called Matrix Product State (\ref{eq:MPS}). Each of the local matrices of Matrix Product States have a covariance matrix of finite dimension.  }
 	\label{fig:MERAapproxMPS}
 \end{figure}
 Keeping these ideas in mind, we elaborate on the following arrangement of tensors pictured in figure (\ref{fig:CIRCUIT}). The output in the horizontal axis is the joint probabilty of the increments of the process at each time $t$. The tensors are contracted along the vertical axis, which we will call the virtual time $\tau$.

We start from a discretized construction. In order to derive a continuum limit, we need to refine the lattice spacing in the real and virtual time with the sample size $N$.
It turns out the following choice of tree tensors yields fractional Brownian motion, 
\begin{align}
P\left(Y^{t_{k}}_{\tau_k}\right.&\left.|Y^{t_{k}}_{\tau_{k+1}}=z_1 ,Y^{t_{k+1}}_{\tau_{k+1}}=z_2\right)\nonumber \\ &\propto  \mathcal{N}\left(y; m_{\tau_k}^{\tau_{k+1}} (z_1+z_2), \sigma^{t_k}_{\tau_{k}} \right) 
\label{eq:tree_transition}
\end{align}
\begin{figure}[t!]
	\includegraphics[width=0.4\textwidth]{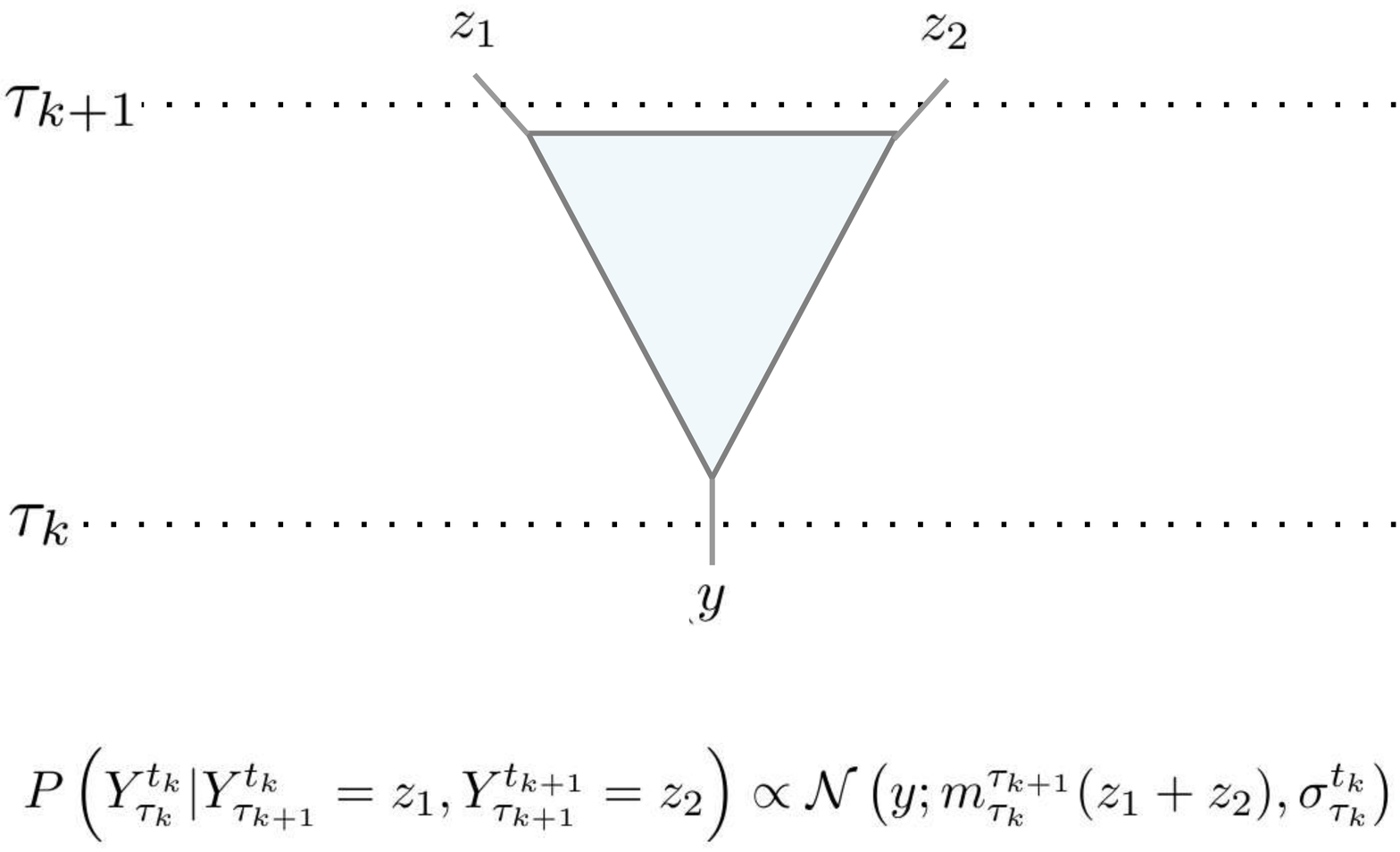}
	\caption{Tensor representation of the gaussian conditional probabilities used for the construction of (\ref{fig:CIRCUIT}).}
	\label{fig:transferCircuit}
\end{figure}
These components are illustrated graphically in figure (\ref{fig:transferCircuit}).
with the affine parameter $m_{\tau_k}^{\tau_{k+1}}$, 
$$m_{\tau_k}^{\tau_{k+1}} = \gamma_k^{k+1}\exp\left([1-H]\int_{\tau_k}^{\tau_{k+1}}ds\ \frac{1}{s}\right),~~ k>0$$
The normalization coefficient $\gamma_n^{n+1}$ is given by,
$$ \gamma_n^{n+1} = \frac{\left[(n+1)\left(\begin{array}{c}2(n+1) \\n\end{array}\right)\right]^{-1/2}}{\left[n\left(\begin{array}{c}2n \\n\end{array}\right)\right]^{-1/2}}$$
and the variance $\sigma_{\tau_k}$
,
$$ \sigma_{\tau_k} = \frac{1}{2}\sqrt{H|2H-1|\Gamma(1-H)^{-1}}\  \sigma$$

We introduce a cutoff on the lower boundary which depends on $\epsilon$,
$$m_{\tau_0}^{\tau_{1}} = \frac{1}{2}\exp\left([1-H]\int_{\epsilon^{1/|1-H|}}^{\tau_{1}}ds\ \frac{1}{s}\right),~~ k>0$$
with variance,
\begin{align*}
\sigma_{\tau_0} = \left\{\begin{array}{cc} \sqrt{\epsilon} & H = 1/2 \\
\frac{1}{2}\sqrt{H|2H-1|\Gamma(1-H)^{-1}}\  \sigma & \mbox{ else} \end{array}\right.
\end{align*}
The time interval $[0,T]$ is divided in $N$ segments of length $\epsilon$. Simultaneously, the virtual dimension is discretized, however, not uniformly,
$$\tau_n = \sqrt{n \epsilon^2} $$
 In the limit, $N\to \infty$ and $\epsilon\to 0$, we show that this joint probability distribution indeed represents fractional Brownian motion. Further details are given in the appendix (\ref{app:fbm}).

\begin{theorem}
The process $B^{\mbox{\scriptsize fm}}_T = \sum_j X_{t_j}$ with joint probalility distribution $P(X_{t_1},...,X_{t_N})$ constructed earlier, is a fractional brownian motion with Hurst index $H$ in the limit $N\to \infty$ and $t_N \to T$.
\end{theorem}
\begin{proof}
	Since the process is gaussian, it is sufficient to show that,
$$ E[B^{\mbox{\scriptsize fm}}_t B^{\mbox{\scriptsize fm}}_s]= \frac{1}{2}\sigma^2 \left( t^{2H} + s^{2H}-|t-s|^{2H}\right),~~ s<t<T$$
One should see that the following relation for $H\not = 1/2$,
\begin{align*}
\frac{1}{2}\left( t^{2H} + \right.&\left.s^{2H}-|t-s|^{2H}\right) \\&= H (2H-1)\int_0^t du_{1}\int_0^s du_2 |u_1-u_2|^{2H-2}
\end{align*}
Using this relation and noticing that,
$$E\left(X_{t_j}X_{t_k}\right) \approx \sigma \epsilon^2H (2H-1) |t_j-t_k|^{2H-2} $$
the claim follows. 
\end{proof}
\begin{figure}[t!]
	\vspace{-3 cm}
	\subfloat{\includegraphics[width = 0.4\textwidth]{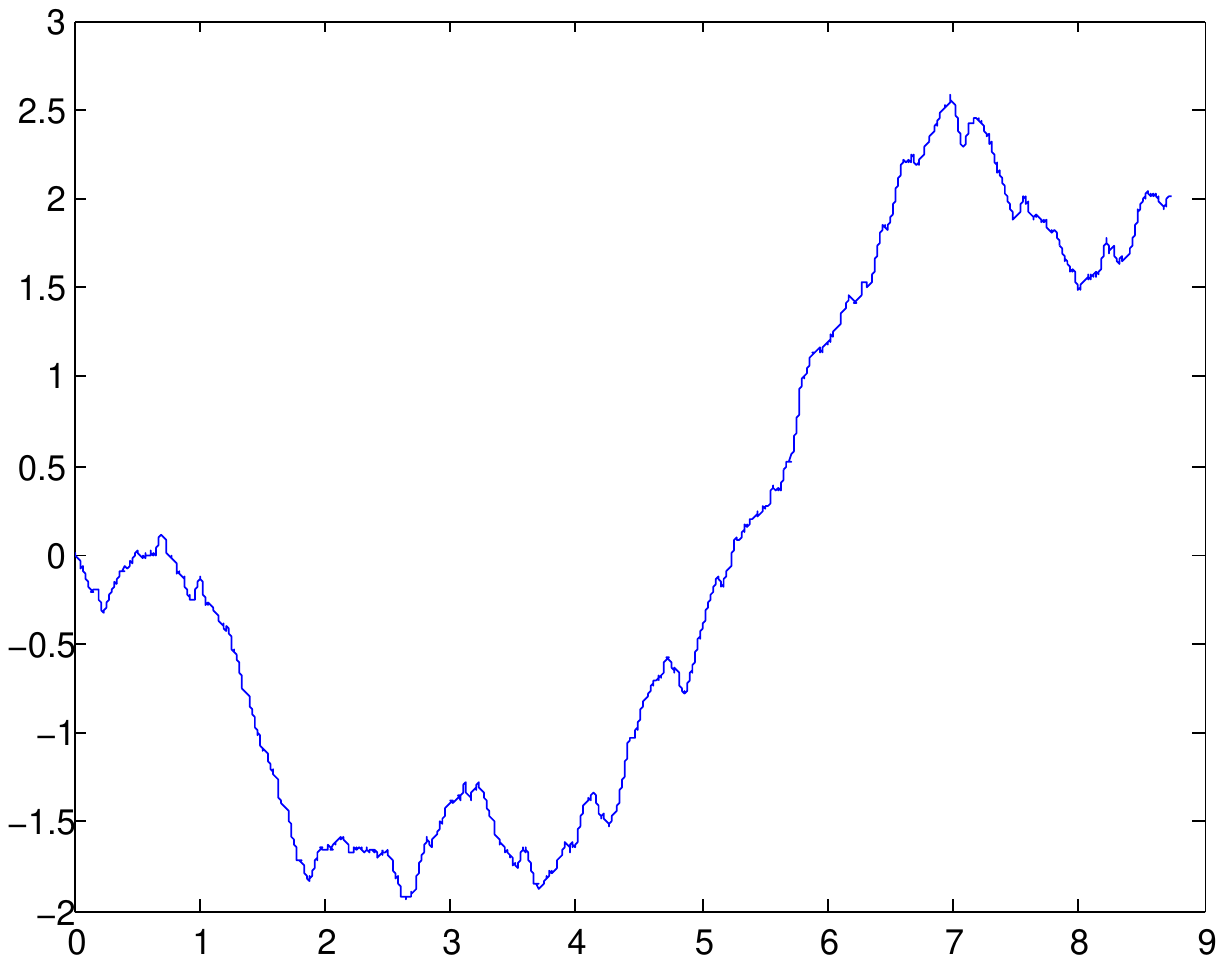}} \\
	\vspace{-7 cm}
	\subfloat{\includegraphics[width = 0.4\textwidth]{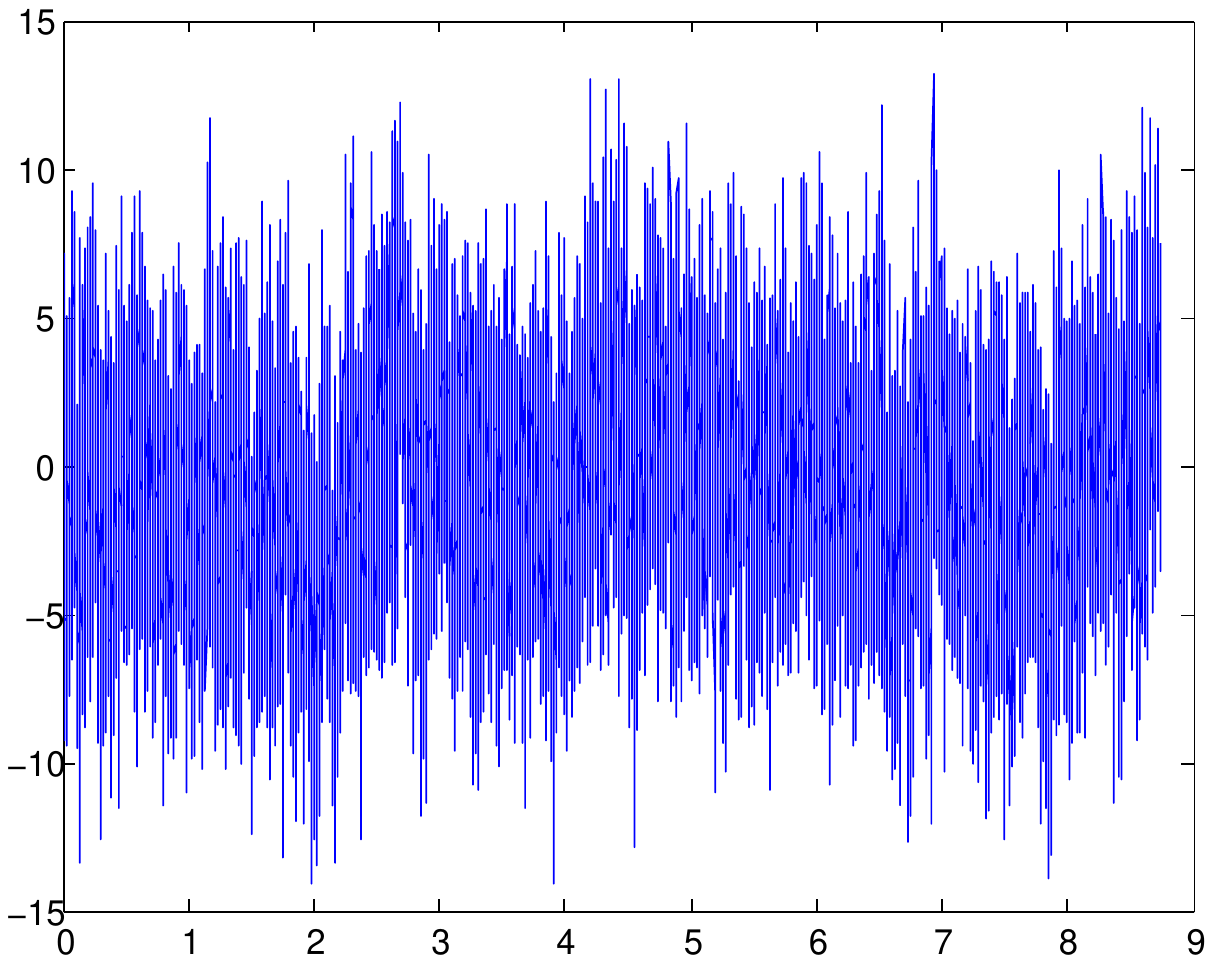}}
	\vspace{-3 cm}
	\caption{Plot of the path and increments of a fractional Brownian motion generated using MERA.}
	\label{fig:path_incr}
\end{figure}
\subsection{MERA: Sampling integrals of fractional Brownian motion}
\label{section:Numerical}

Clearly, the integral representation (\ref{eq:int_fbm}) is very difficult to sample. There exists however a few methods \cite{Dieker_simulation_fbm} for sampling fractional brownian motion. 
An even-more challenging problem is the sampling of an integral with respect to fractional Brownian motion. 
$$\int_0^t f(s)dB^{\mbox{\scriptsize fm}}(s) $$
Networks seem to present an elegant alternative solution to this problem. As we can easily sample $X_{t_j}$, so that,
\begin{equation} 
\label{eq:sampling int}
\sum_j  f(t_j)X_{t_j}\to \int_0^t f(s)dB^{\mbox{\scriptsize fm}}_s 
\end{equation}
converges to the desired result for sufficiently large $N$.
The sampling of the sum in equation (\ref{eq:sampling int}) can be done in the following  way.
As often used in controle theory, the transition tensor $P(Y_{\tau_k}|Z_{\tau_{k+1}}=z)$ in equation (\ref{eq:tree_transition}) can be shown to be equivalent to the algebraic equation,
\begin{equation}
Y^{t_{k}}_{\tau_k}=m_{\tau_k}^{\tau_{k+1}}[t_k](Z^{t_{k}}_{\tau_{k+1}}+Z^{t_{k+1}}_{\tau_{k+1}}) +\xi^{t_{k}}_{\tau_k}
\label{eq:algebraic}
\end{equation}
\begin{figure}
	\vspace{- 3cm}
	\includegraphics[width=0.4\textwidth]{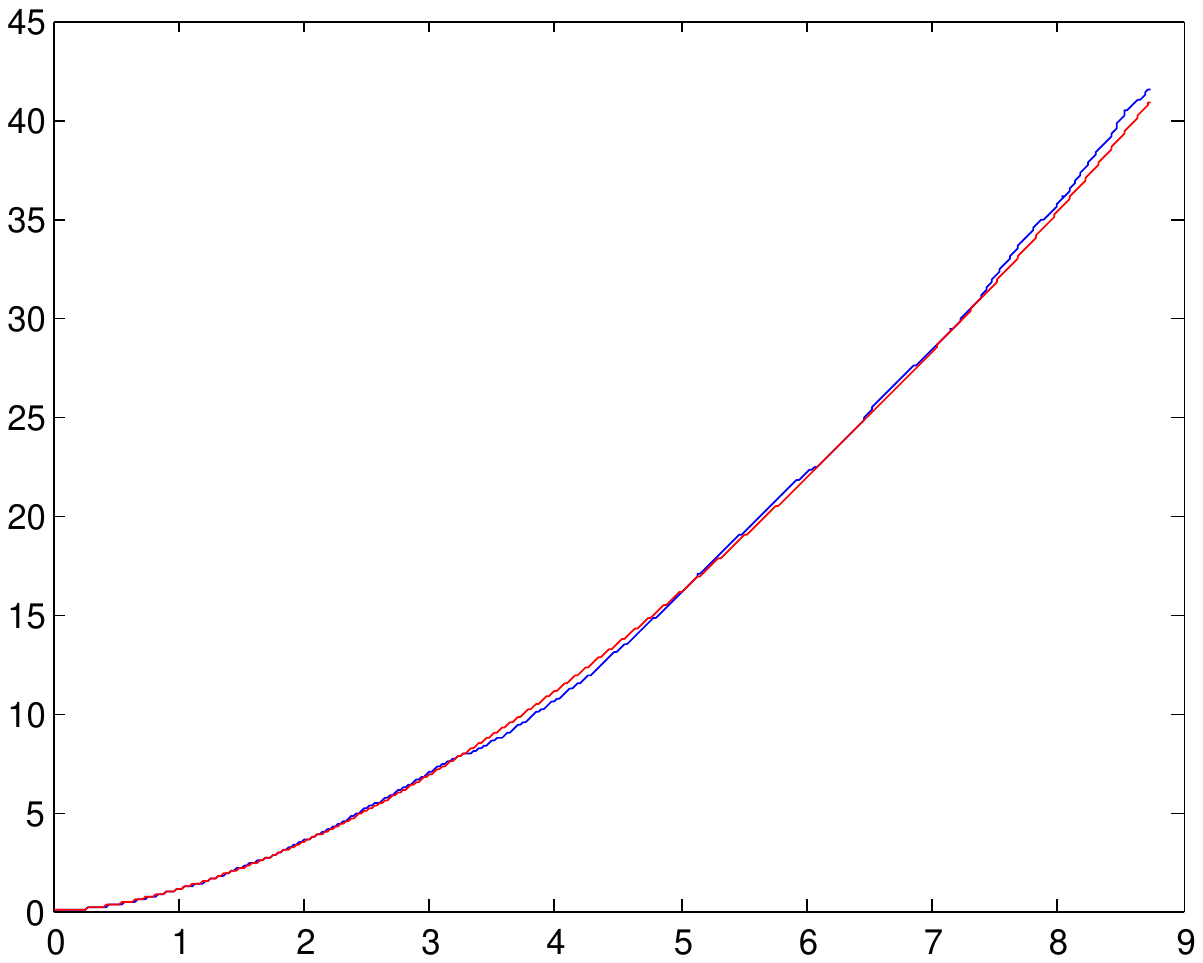}
    \vspace{- 3cm}
	\caption{A fit of the correlation $E(B_t^2)$ for fractional processes generated by a MERA. The correlation is seen to grow as $|t|^{1.66}$.}
	\label{fig:fit}
\end{figure}

 using the additional normally distributed random variable $\xi_{\tau_k}\propto\mathcal{N}(0,\sigma_{\tau_k})$

If we  combine and trace the vector in equation (\ref{eq:algebraic}) with variables $(c_1,c_2)$, we derive a discretized renormalization flow equation,
\begin{align*}
\sum_{j=N_1}^{N_2}c_jY^{t_{j}}_{\tau_k}&= \sum_{j=N_1}^{N_2}c_j \xi^{t_{j}}_{\tau_k}+m_{\tau_k}^{\tau_{k+1}} \sum_{j=N_1+1}^{N_2}(c_{j-1}+c_{j})Y^{t_{j}}_{\tau_{k+1}}\\&+ m_{\tau_k}^{\tau_{k+1}}\left(c_{t_{N_1}}Y^{t_{N_1}}_{\tau_{k+1}}+c_{t_{N_2+1}}Y^{t_{N_2}}_{\tau_{k+1}}\right)
\end{align*}
As we go higher up the tree, we replace the sum of the variables on each branch by a sum over the local fluctuactions and we renormalize the terms which are connected by a higher branch.

For example, we readily derive that for $H\not=1/2$ ,
\begin{align*}
B^{\mbox{\scriptsize fm}}_T \approx  \alpha_{H}\sigma\sum_{n=0}^{N_{\infty}}m_{\tau_0}^{\tau_{n}} \sum_{k=0}^{N+n}\left(\theta\left(a(k \leq \ceil[\Big]{\frac{n}{2}}\right)\left(\begin{array}{c}n\\ k\end{array}\right)\right.\\
\left.+\theta\left(k \geq N+\ceil[\Big]{\frac{n}{2}}+1\right)\left(\begin{array}{c}n\\ k-\left(N+\ceil[\Big]{\frac{n}{2}}+1\right)\end{array}\right)\right.\\
\left. \theta\left(\ceil[\Big]{\frac{n}{2}}+1\leq k \leq N+\ceil[\Big]{\frac{n}{2}}\right)\left(\begin{array}{c}n\\ k\end{array}\right)\right) \xi_{t_k}^{\tau_n} 
\end{align*}
We denoted $N_{\infty}$ as the debt of the circuit. Naturally as proven this should be as large as possible, $N_{\infty}\to \infty$. It turns out that such finite debt circuit are related to so-called Matrix Product States which we discuss in the next section.
Unfortunately, the circuit (\ref{fig:CIRCUIT}) presented earlier is too slow, $O(N^4)$.
It seems however that MERA appears as a powerful tool. The main feature of MERA is the renormalization of the real space into the virtual space. This geometry reduces the complexity to the order of $O(N\log N)$. A basic method for simulating gaussian processes is by Cholesky decomposition which is of the order $O(N^3)$. Other more powerful methods scale as $O(N^2)$. 
In figure (\ref{fig:path_incr}), we have plot the path and increments of a fractional motion. We calibrated the MERA by approximating the covariant matrix of the process. In figure (\ref{fig:fit}), we have plotted a fit of the correlation $E(B_t^2)$. This correlation is evaluated by Monte Carlo and generating the process with MERA.

\begin{figure}
	\includegraphics[width=0.4\textwidth]{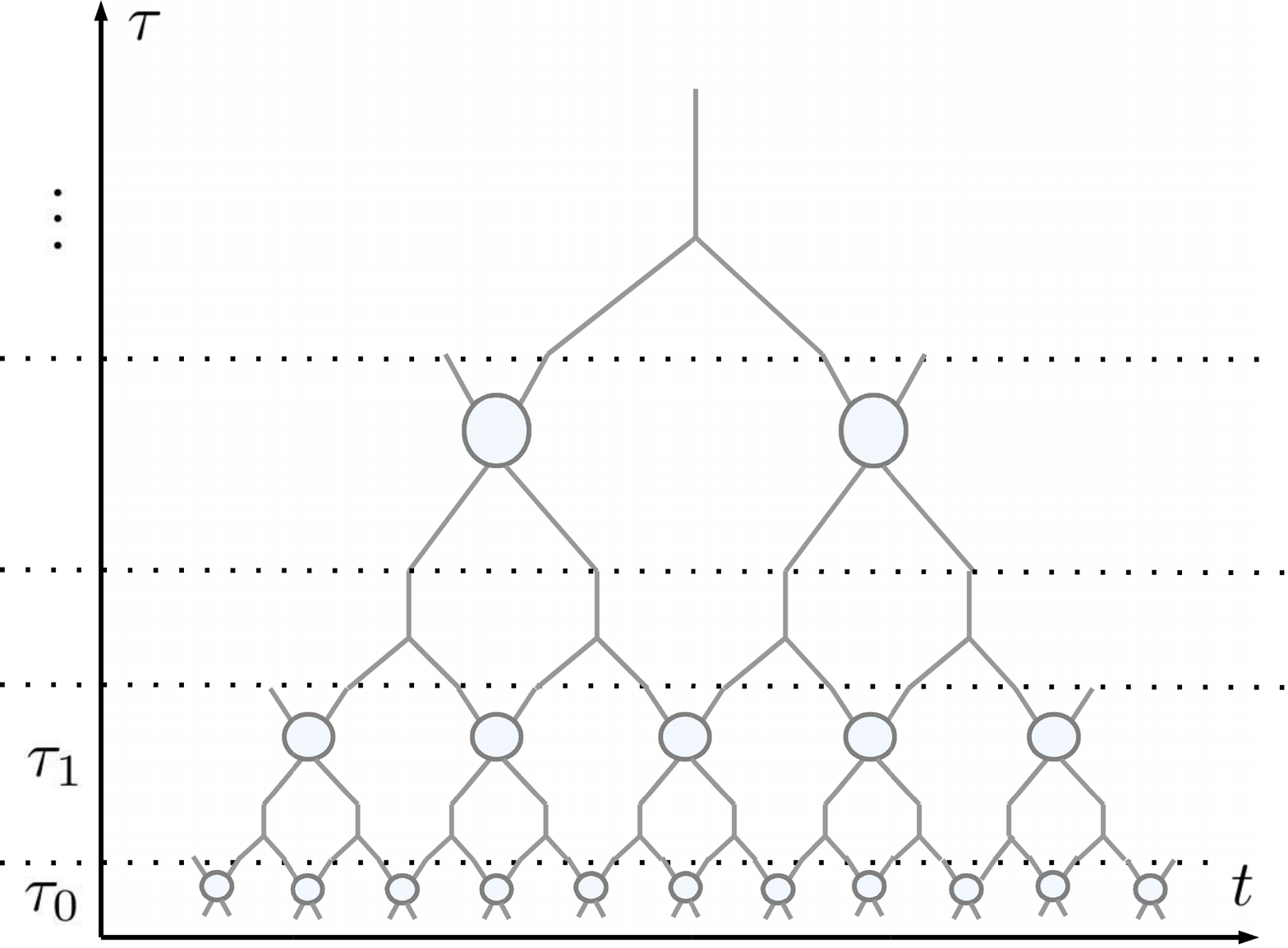}
	\caption{The multi-scale entanglement renormalization ansatz. The end of the circuit represents a quantum states, and in our paper the joint probability distribution of a process.}
	\label{fig:MERA}
\end{figure}
\subsection{Matrix Product State Representation}
Matrix Product States were originally understood as the ansatz for the density renormalization group algorithm \cite{White_1992_DMG}. The construction of these states has appeared, disappeared and reappeared many times through history under many different names such as finitely correlated states, Tensor trains, (complex) (quantum) hidden markov chains, in many different fields such as data science, quantum physics, statistical mechanics,... By consequence, we will focus on the form of interest for this paper. For a gaussian process $B_{T=t_N} = \sum_{j=1}^N X_{t_j}$, one may want to rewrite the joint probability distribution as follows,
\begin{align}
\label{eq:MPS}
&P(X_{t_1}=x_1,...,P(X_{t_N}=x_N))\nonumber \\&= \int_{\mathbb{R}^N}d\vec{u} A^{(x_2)}(u_2,u_3)...A^{(x_{N-1})}(u_{N-1},u_N)A^{(x_{N})}(u_N)
\end{align}

The Matrix Product State tensor $A^{(x_j)}(u_{j},u_{j+1})$ is graphically represented in figure (\ref{eq:MPS}). If we decide to cut the circuit up to a height $N_{\infty}<\infty$, this circuit can be represented by a Matrix Product State as shown in figure (\ref{fig:MERAapproxMPS}).
It only rests us to precisely evaluate how the debt of the circuit $N_{\infty}$ and the sample size are related to some error $\delta$. The sample size should imply an error of at leasy the order $O(1/N)$. Furthermore, error due debt the circuit depends on the affine parameter $m_{\tau_0}^{\tau_k}$ which decays as $O((1/\tau)^{1-H})$. The quantification of the error is tricky. Ideally, we should convergence look at the convergence in the $1$-norm $\|.\|_1$ of the distributions. However, this is analytically not feasible. We could instead compare the covariance matrices. It seems the easiest to study the following error $\delta$,
\begin{align}
\delta = \max_{0\leq s\leq t} \Big| E\left(B^{\mbox{\scriptsize fm}}_s B^{\mbox{\scriptsize fm}}_t\right)-\sum_{n_1,n_2=0}^{N_1,N_2}E\left(X_{t_{n_1}}X_{t_{n_2}}\right)\Big|
\label{eq:error}
\end{align}
with $t_{N_1},t_{N_2}\to s,t$.
The covariance of a gaussian process with $N$ increments consists of at most $N^2$. Clearly the circuit (\ref{fig:MERAapproxMPS}) is too "deep" as it contains at least $N^4$ parameters in some approximation. However, MERA suggest a reduction to $O(N \log N)$ parameters is possible.

\begin{figure}[t!]
	\includegraphics[width=0.4\textwidth]{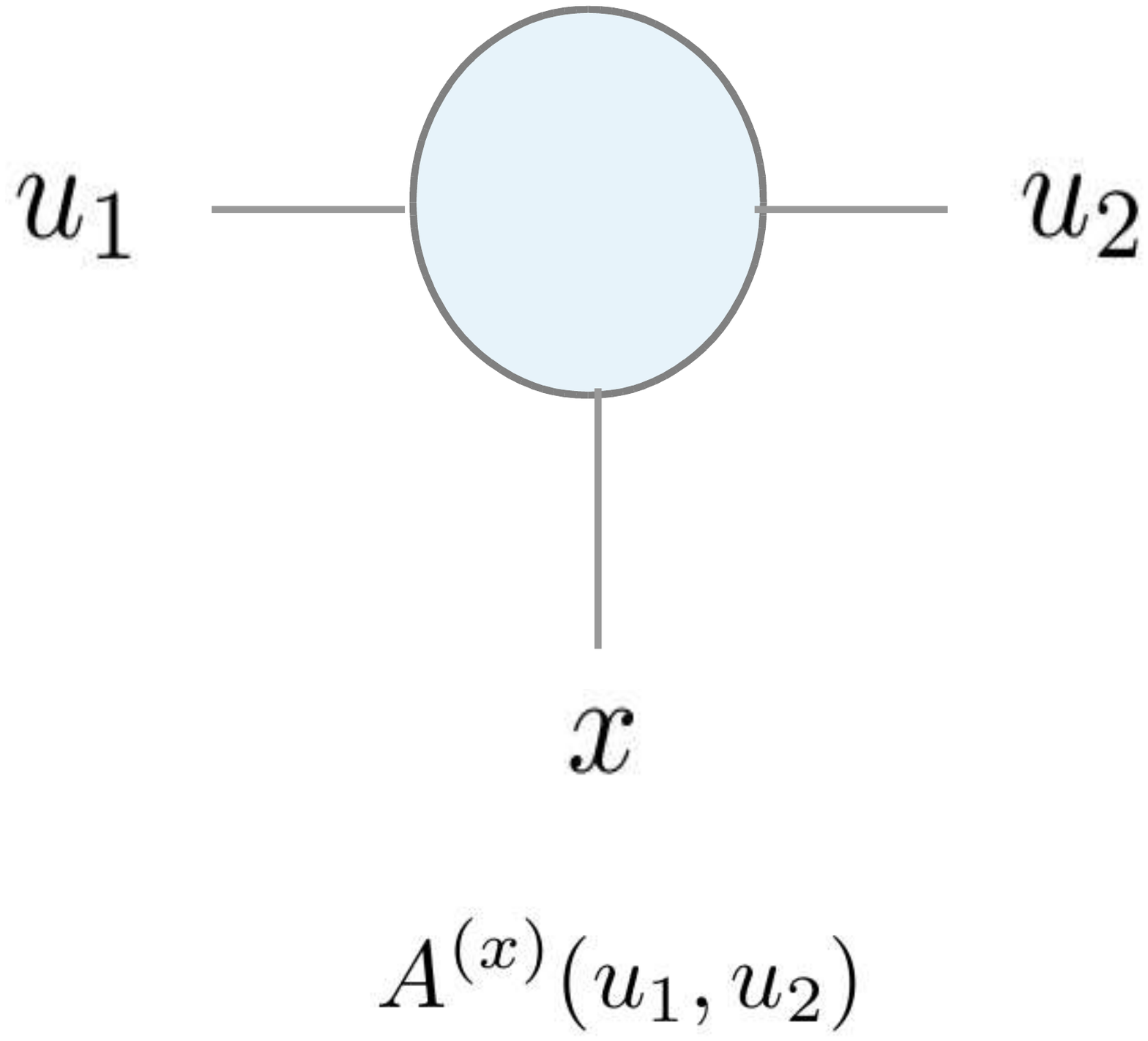}
	\caption{Matrix Product States tensor  }
	\label{fig:MPStensor}
\end{figure}
\section{Multifractal Properties}

Fractional Brownian motion is called unifractal. This property is coupled with  the Hurst index $H$,
$$B^{\mbox{\scriptsize fm}}(a t))\disteq a^{H}B^{\mbox{\scriptsize fm}}(t))$$

A more general property is multifractality. Rather than satisfying a global scaling with a unique affine parameter, there could be a distribution of many local scaling,
\begin{align}
|X(t+a\Delta t))-X(t)|\disteq M(a)|X(t+a\Delta t))-X(t)|
\label{eq:multifractality}
\end{align}
This time, however, $M(a)$ is a positive random variable which only depends on $a$ and not $t$.
This is achieved from our construction by the introduction a randomization in the Hurst index $H$ inside the now-random variable $m_{\tau_k}^{\tau_{k+1}}$, 
\begin{align}
m_{\tau_k}^{\tau_{k+1}} = \gamma_k^{k+1}\exp\left(\log M(\tau_{k+1})-\log M(\tau_{k})\right)
\label{eq:multifracCOND1}
\end{align}
The one-parameter random variable $M(\tau)$ satisfies the additional multiplicative property,
\begin{align}
M(a \tau) \disteq M_1(a) M_2(\tau)
\label{eq:multifracCOND2}
\end{align}
with $M_1$ and $M_2$ independent random variables. More details and examples are given in the appendix (\ref{app:multifrac}). We can use the structure (\ref{fig:CIRCUIT}) to derive the following extension of the previous theorem to multifractal processes.
\begin{figure}[t!]
	\includegraphics[width=0.4\textwidth]{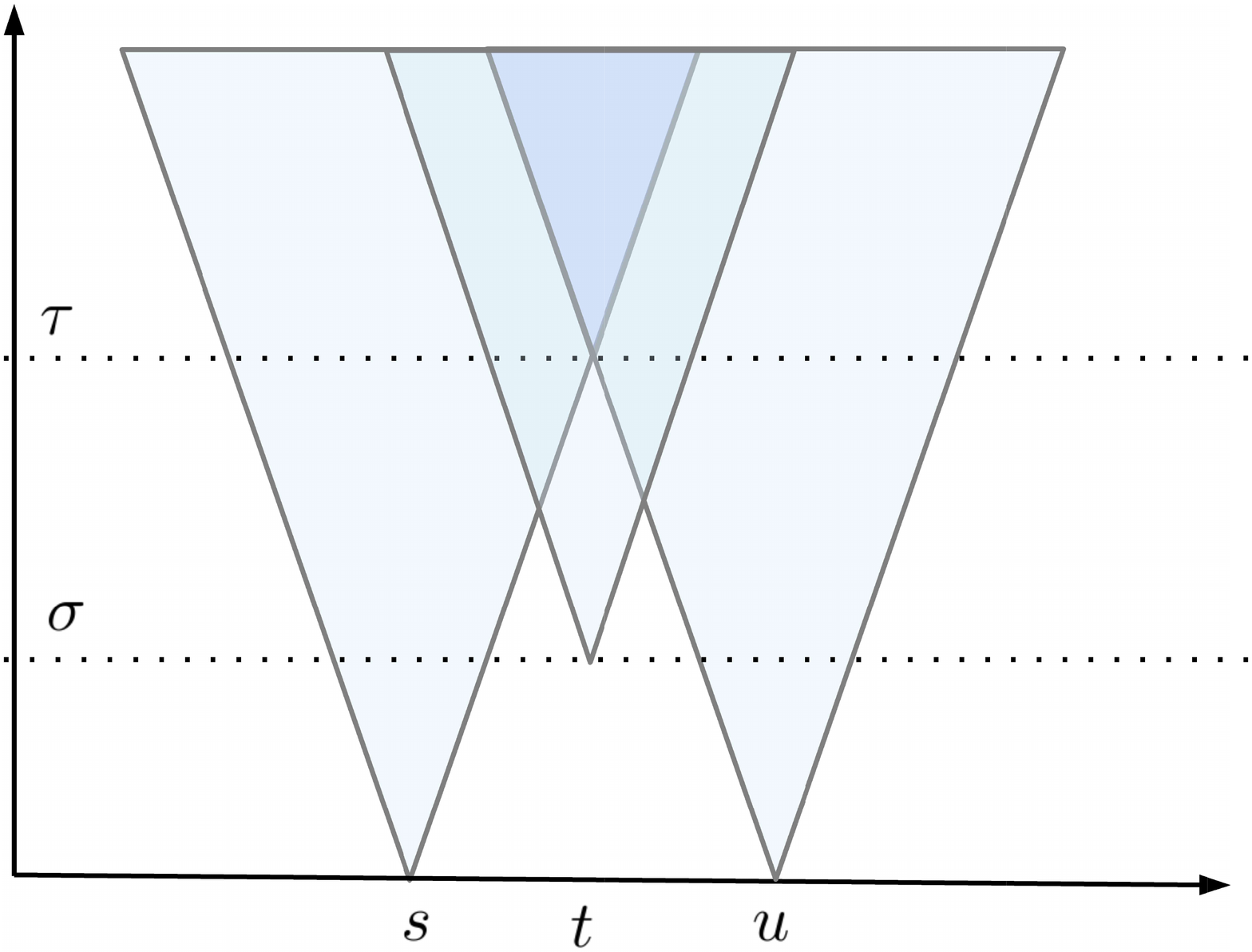}
	\caption{The increments $X_{s}$ and $X_{u}$ depend solely on random variables inside lightcones starting respectively at $s$ and $u$. Their correlation $g(\tau)$ is then determined by the overlap of the lightcones at $(t,\tau)$. Per construction, the correlation sastifies a renormalization property as $g(\tau)=\left(\frac{\tau^2}{\sigma^2}\right)^{2H-2}g(\sigma)$.}
	\label{fig:lightcone}
\end{figure}
\begin{theorem}
The joint probability distribution of the process $X(t) = \sum_{j=1}^N X_{t_j}$  constructed using $m_{\tau_k}^{\tau_{k+1}}$ as given by equation (\ref{eq:multifracCOND1}) and with random variables $M_j(\tau)$ satisfying property (\ref{eq:multifracCOND2}) implies the local scaling (\ref{eq:multifractality}) and multifractality,
	$$X(at) \disteq M(a) X(t) $$
\end{theorem}
\begin{proof}
	Similarly to the proof for the fractional Brownian motion, it is sufficient to find scaling on the level of the correlations $E\left(X_{t_k}X_{t_l}\right)$.
	The key intuition is illustrated in figure (\ref{fig:lightcone}). Both $X_{s=t_k}$ and $X_{u=t_l}$ depend on random variable inside the light cones $s$ and $u$ respectively. Hence the correlations is determined by the random variable inside their intersection which is the light cone starting at $(t=|u-s|/2,\tau)$. In other words these random variable determine a new random variable $X_{t,\tau}$. Scaling is implied if,
	$$X_{t,a\tau} \disteq M(a) X_{t,\tau}  $$
	This is precisely implied by construction of $m_{\tau_k}^{\tau_{k+1}}$ and $M$ in equations (\ref{eq:multifracCOND1}) and (\ref{eq:multifracCOND2}). Further technical details can be found in the appendix (\ref{app:multifrac}).
\end{proof}

\section{Conclusion and Further Directions}
In this work, we constructed the joint probability measure of the increments of fractional brownian motion using the framework of tensor network. This insight presented us on one hand with a new sampling method of fractional Brownian motion using the already known MERA, used in the study of quantum critical systems. Secondly, we show that such circuits present a novel pictorial representation of multifractal processes. In such representations, the connection between multifractality and renormalization emerges naturally.
This network representation also presents a new insight of the multifractal processes. In the language of networks, the Hurst index is not unique anymore but the value on the different levels of the circuit is sampled from a self-similar measure.

\section*{Acknowledgements}
We thank Jutho Haegeman for helpful discussions.
We acknowledge financial support by the FWF project CoQuS No.  W1
210N1 and project QUTE No. H20ERC2015000801.

\bibliographystyle{plain}

\begin{widetext}
\newpage
\appendix
\section{Fractional Brownian Motion}
\label{app:fbm}
\begin{theorem*}
	Given the join probability distribution
	\begin{align}
	P(\vec{X}_{t_j,\tau_0})=\sum_{j_1,...,j_{N_{\infty}}}P(\vec{X}_{t_{j_0},\tau_0}|\vec{X}_{t_{j_1}},\tau_1)P(\vec{X}_{t_{j_1}},\tau_1|\vec{X}_{t_{j_2}},\tau_2)...P(\vec{X}_{t_{j_{N_\infty-1}}},\tau_{N_{\infty-1}}|\vec{X}_{t_{j_{N_\infty}}},\tau_{N_{\infty}})
	\label{eq:joint_distr}
	\end{align}
	which is constructed from the network whose structure is pictured in figure (\ref{fig:MERAapproxMPS})
	\begin{align}
	P(\vec{X}_{t_{j_k},\tau_k}|\vec{X}_{t_{j_{k+1}}},\tau_{k+1})=\prod_l P(X_{t_l},\tau_k|X_{t_l},X_{t_{l+1}},\tau_{k+1})
	\label{eq:transfer}
	\end{align}
	with the transfer operations given by,
	\begin{align}
	P(X_{l,\tau_k}|X_{t_l,\tau_{k+1}}=z_1,X_{t_{l+1},\tau_{k+1}}=z_2)\propto  \mathcal{N}\left(y; m_{\tau_k}^{\tau_{k+1}} (z_1+z_2), \sigma^{t_k}_{\tau_{k}} \right) 
	\label{eq:tree_transition_append}
	\end{align}
	The parameters are taken to be,
	$$t_n=\epsilon n,~\tau_n = \sqrt{n \delta},~\delta = \epsilon^2,~m_{\tau_j}^{\tau_k}= \left(\frac{\tau_k}{\tau_j}\right)^{H-1}\frac{\left[ k\left(\begin{array}{c}2k \\k\end{array}\right)\right]^{-1/2}}{\left[ j\left(\begin{array}{c}2j \\j\end{array}\right)\right]^{-1/2} },~~\sigma^{t_k}_{\tau_{k}}=\frac{1}{2}\sigma \sqrt{H|2H-1|\Gamma(1-H)^{-1}} $$
	The variance $\sigma^{t_k}_{\tau_{k}}$ satisfies the boundary condition,
	\begin{align*}
	\sigma_{\tau_0}^{t_0} = \left\{\begin{array}{cc} \sqrt{\epsilon}\sigma & H = 1/2 \\
	\frac{1}{2}\sigma \sqrt{H|2H-1|\Gamma(1-H)^{-1}}   & \mbox{ else} \end{array}\right.
	\end{align*}
	 In the limit $N\to \infty$ and $t_N \to T$, the process $B^{\mbox{\scriptsize fm}}_T = \sum_j X_{t_j}$ with joint probalility distribution $P(X_{t_1},...,X_{t_N})$, is then a fractional brownian motion with Hurst index $H$, which satisfies,
	$$ E\left(B_{t_{N_1}}B_{t_{N_2}}\right)=\frac{1}{2}\sigma^2\left( t_{N_1}^{2H}+t_{N_2}^{2H}-|t_{N_1}-t_{N_2}|^{2H}\right)$$
\end{theorem*}

\begin{proof}
	The calculations are more insightful if we keep figure (\ref{fig:lightcone}) and the renormalization flow equation (\ref{eq:algebraic}) in mind. The correlation is then determined by the overlap of the lightcones of $X_{t_k}$ and $X_{t_l}$, which is another lightcone at $(|t_k -t_l|/2,\tau_n=\sqrt{|t_k -t_l|\epsilon/2})$.
	One can check that,
	\begin{align}
	E(X_{t_k}X_{t_l}) & \propto \frac{1}{2}\sigma^2 \sum_{q=n}^{N_{\infty}}\left(m_{\tau_0}^{\tau_q}\right)^2\sum_{j=n}^q\left(\begin{array}{c}q \\j\end{array}\right)\left(\begin{array}{c}q \\j-n\end{array}\right)\nonumber \\
	&= \frac{1}{2}\sigma^2 \epsilon^2 \sum_{q=n}^{N_{\infty}}\tau_{q}^{2H-2} \left[\sum_{j=n}^q\left(\begin{array}{c}q \\j\end{array}\right)\left(\begin{array}{c}q \\j-n\end{array}\right)\right]\left[ q\left(\begin{array}{c}2q \\q\end{array}\right)\right]^{-1}\label{eq:sum_corr}
	\end{align}
Using Vandermonde Convolution's identity and Stirling's formula, we can approximate the binomial coefficients,
	$$ \sum_{j=n}^{q} \left(\begin{array}{c}q \\j\end{array}\right)\left(\begin{array}{c}q \\j-n\end{array}\right) = \left(\begin{array}{c}2q \\n\end{array}\right),~~\left(\begin{array}{c}2q \\q-n\end{array}\right)\left(\begin{array}{c}2q \\q\end{array}\right)^{-1}\approx e^{-n^2/q}$$
	
	Introducing a rescaling $\gamma(q)=q/n^2$, and using the identity followed by approximation above simplifies the equation (\ref{eq:sum_corr}),
	\begin{align*}
	\sum_{q=n}^{n^2} \tau_{n}^{2H-2} \left(\begin{array}{c}2q\\ q - n \end{array}\right) \left(\begin{array}{c}2q\\ q \end{array}\right)^{-1}=\left(\sum_{q=n}^{N_{\infty}}n^{-2}\exp\left(-\gamma(q)^{-1}\right)\gamma(q)^{H-2}\right)t_{n}^{2H-2}
	\end{align*}
	In the limits $n^2/N_{\infty} \to 0$ and $n\to \infty$, the expression between brackets converges to the Gamma function,
	$$\sum_{q=n}^{N_{\infty}}n^{-2}\exp\left(-\gamma(q)^{-1}\right)\gamma(q)^{H-2}\approx \int_{n^2/N_{\infty}}^ndu\ e^{-u}u^{-H}\to \Gamma(1-H)$$
Combining the results yields the correlation,
\begin{align*}
E(X_{t_k}X_{t_l}) &\propto\frac{1}{2} \sigma^2 \epsilon^2 t_{|k-l|}^{2H-2}
\end{align*}
For large $N$ and small $\epsilon$, we can approximate the double sum by a double integral,
$$\sum_{j_1=1}^{N_1}\sum_{j_2=1}^{N_2}\epsilon^2\tau_{|k-l|}^{2H-2} \approx \int_0^{t_{N_1}}du_1\int_0^{t_{N_2}}du_2 |u_1-u_2|^{2H-2} =\frac{1}{2} \sigma^2\left( t_{N_1}^{2H}+t_{N_2}^{2H}-|t_{N_1}-t_{N_2}|^{2H}\right)$$
from which the claim follows.
\end{proof}

\section{Multifractal process}
\label{app:multifrac}
A key component for the introduction of multifractal measures, are the so-called self-similar measures. A detailed introduction can be found in \cite{Mandelbrot_multifrc_asset,Riedi_multifrc}. Define the set $\mathcal{S}$ to consist of all similitude transformation, i.e. translation and homothetic transformations.
\begin{definition}
Given $\mu:[0,T]\to [0,1]$ a random measure, which satisfies,
\begin{enumerate}
	\item For all similitudes $S \in \mathcal{S}$, for any interval $I_1 \subset I_2$, the ratios,

\begin{align}
\label{def:prop_self_similar_1}
\frac{\mu(SI_1)}{\mu(SI_1)} \disteq \frac{\mu(I_1)}{\mu(I_1)}
\end{align}
are equal in distribution as long as $I_1,I_2, S I_1, SI_2 \subset [0,T]$.

	\item For all decreasing sequences of compact intervals $I_1\subset I_2 \subset... \subset I_n \subset [0,T]$, the ratios,
	
	\begin{align}
	\label{def:prop_self_similar_2}
	\frac{\mu(I_1)}{\mu(I_2)},...,	\frac{\mu(I_{n-1})}{\mu(I_n)}
	\end{align}
	are statistically independent.
\end{enumerate}
then, the measure $\mu$ is called self-similar.
\end{definition}
The first property (\ref{def:prop_self_similar_1}) implies the existence of a random variable $M$  such that,
\begin{align}
\label{def:prop_M_1}
\mu[0,ct]\disteq M(c)\mu[0,t],~~0 \leq ct,t \leq T 
\end{align}
From the second property (\ref{def:prop_self_similar_1}), we also derive that the random variable must satisfy a multiplicative property. Taking $0 \leq c_1,c_2\leq 1$ and $ 0\leq t \leq T$, 
$$\frac{\mu[0,c_1 c_2 t]}{\mu[0,t]}= \frac{\mu[0,c_1 c_2 t]}{\mu[0,c_2 t]}\frac{\mu[0,c_2 t]}{\mu[0,t]}$$
Hence, by the corrolary of the first property of self-similar measures (\ref{def:prop_M_1}), 
$$M(c_1c_2)\disteq M'(c_1)M''(c_2) $$
The second property (\ref{def:prop_self_similar_1}) implies that $M'$ and $M''$ are independent.
The existence of such random variable plays a central role when introducing multifractality.
Before jumping to the derivation of our result, we illustrate this property with two examples.
\begin{example}
	For $t>1$, define $M(t)$,
	$$ M(t) = \exp\left(\sigma B_{\log(t)}-\frac{\sigma^2}{2}\log(t) \right) $$
	with the brownian motion $B_{\log(t)}=\int_0^{\log(t)} dB_s$. From the indepence of the increments of Brownian motion, we readily derive,
	$$ M(c t) = \exp\left(\sigma B_{\log(c)}-\frac{\sigma^2}{2}\log(c) \right)  \exp\left(\sigma B_{\log(t)}-\frac{\sigma^2}{2}\log(t) \right)= M(c)M(t) $$
\end{example}
\begin{figure}[t!]
	\includegraphics[width=0.5\textwidth]{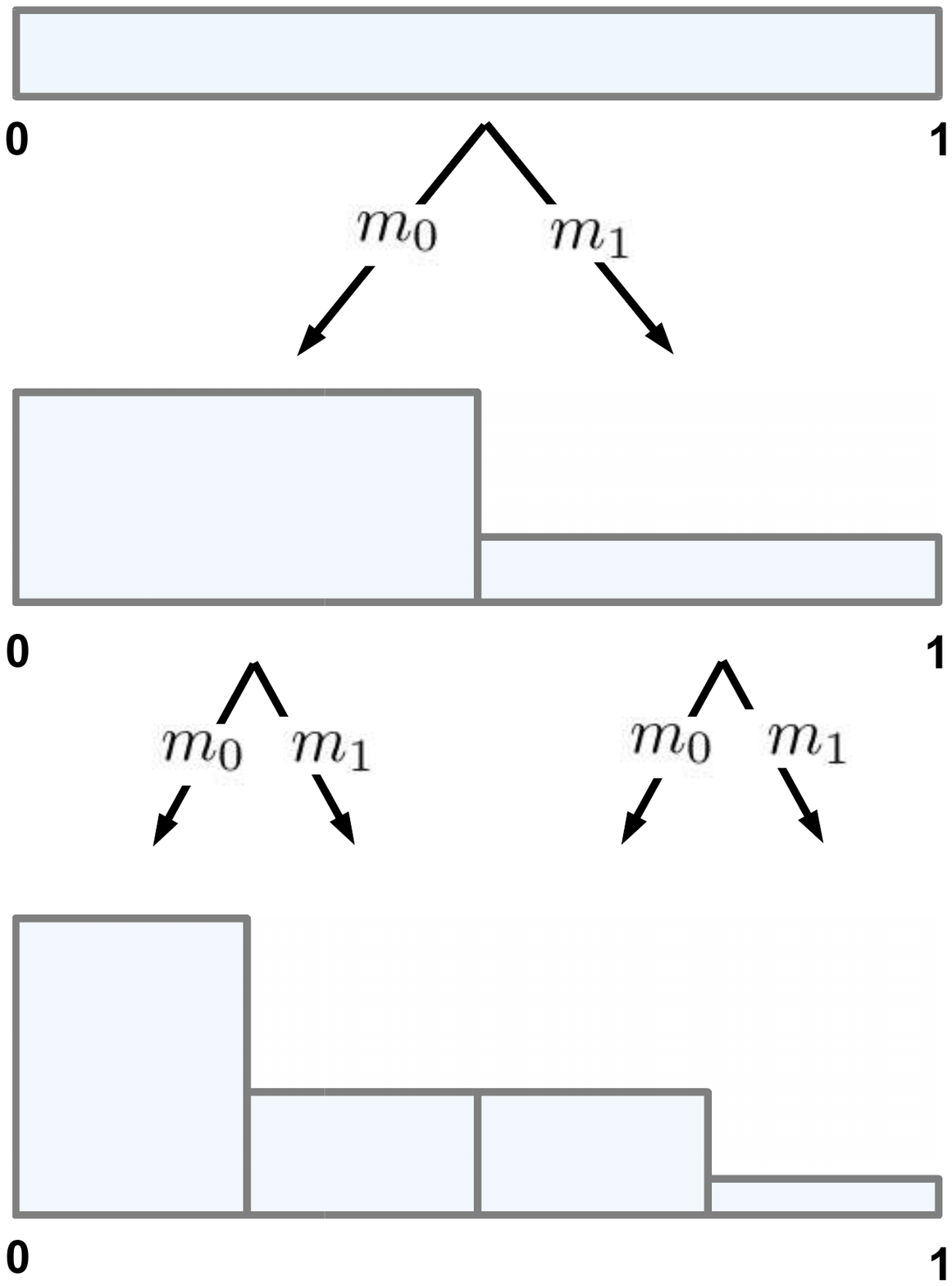}
	\caption{Illustration of the multiplicative cascade for constructing the binomial multifractal measure in example (\ref{ex:binomial}). In the second step, the masses of each cell are respectively $m_0^2, m_0 m_1, m_0 m_1$ and $m_1^2$.}
	\label{fig:cascade}
\end{figure}
\begin{example}[Binomial Measure]
	
	As we will see this example does not satisfy the properties for all $c,t$, but is, however, very insightful.
	The binomial measure is introduced as the limit of an elementary iterative procedure called a multiplicative cascade.
	As illustrated in figure (\ref{fig:cascade}), the idea is to iteratively divide the interval, for sake of simplicity $[0,1]$, into $b$-adic cells of length $1/b$. At each step, the mass is multiplied by a factor depending on the location of the cell.
	For example after $k$ steps, the mass of the cell $t=\sum_{j=1}^k \eta_j 2^{-j}$ with length $\Delta t = 2^{-k}$ is,
	$$\mu[t,t+\delta t] = M(\eta_1)M(\eta_1,\eta_2)...M(\eta_1,...,\eta_k)$$
	Additionally, we can choose each $M(\eta_1,...,\eta_k)$ to be independent.
	
	The most simple case is to choose a unique weight $m_0$. At each step of the iteraction one multiplies by $m_0$ if we take the left cell and $m_1=1-m_0$ for the right cell.
	By induction for $t_{n+1}=t_n + \eta_{n+1} 2^{-(n+1)}$ and $\Delta_n = 2^{n}$,
	$$\mu[t_{n+1},t_{n+1}+\Delta t_{n+1}] =\mu[t_{n},t_{n}+\Delta t_{n}] \left\{\begin{array}{cc}m_0 & \mbox{if } \eta =0 \\ m_1 & \mbox{if } \eta =1 \end{array}\right.$$
	Let us repeat the procedure $N>>1$ times,  and consider two dyadic numbers $c$ and $t$. One should see that by the self-similarity of the construction, the multiplicative property follows,
	$$\mu[ct,ct+\Delta t_{N}] = \mu[c,c+\Delta t_{N}] \mu[t,t+\Delta t_{N}]$$
	Hence, in this example we can construct,
	 $$M(t)=\lim_{N\to \infty}\mu[t,t+\Delta t_{N}]$$
	 
	More information about this construction can be found \cite{Evertszy_1992_Multifractal}.
	
\label{ex:binomial}
\end{example}
\begin{theorem}
	Given the join probabilty distribution with the same structure as in equation (\ref{eq:joint_distr}) and transfer operations (\ref{eq:transfer}) with parameters,
	\begin{align*}
	t_n=\epsilon n,~\tau_n = \sqrt{n \delta},~\delta = \epsilon^2,~m_{\tau_0}^{\tau_n}=\exp\left(\frac{1}{2}\log\left[\frac{ \tau_n^{-4}M(\tau_n)}{\epsilon}\right]\right)\left(\begin{array}{c}2n \\n\end{array}\right)^{-1/2}\\
	m_{\tau_k}^{\tau_l}=\exp\left(\frac{1}{2}\log\left[\frac{ \tau_l^{-4}M(\tau_l)}{ \tau_k^{-4}M(\tau_k)}\right]\right)\frac{\left[k\left(\begin{array}{c}2k \\k\end{array}\right)\right]^{-1/2}}{\left[j\left(\begin{array}{c}2j \\j\end{array}\right)\right]^{-1/2} },~~~~\sigma^{t_k}_{\tau_{k}}=\sigma
	\end{align*}
	The continuous one-parameter random variable $M(\tau_n)$ satisfies the multiplicative property,
	\begin{align}
	M(c\tau_n)\disteq M'(c)  M''(\tau_n)
	\label{eq:multiplicative}
	\end{align}
	where $M'(c)$ and $ M''(\tau_n)$ are independent.
	The process $Y_T = \sum_j X_{t_j}$ with joint probalility distribution $P(X_{t_1},...,X_{t_N})$, is then a multifractal process which satisfies the scaling property,
	$$ E\left(Y^m(cT)\right)= E\left(M^m(c)\right)E\left(Y^m(T)\right)$$
\end{theorem}

\begin{proof}

We show for all $m\geq 1$,
$$E(Y_{ct}^m)=E(M(c)^m)E(Y_{t}^m)$$
Let us first fix the value of the random variable $M(.)$ evaluated at different times. The process is then gaussian for all such values. Hence, the moments are zero for odd $m$ and powers of the variance for $m$ even.
Similary to the case of fractional Brownian motion, we evaluate the second moment.
One can check that,
\begin{align}
\tilde{E}(X_{t_k}X_{t_l}) &= \sigma^2 \epsilon^2 \sum_{q=n}^{N_{\infty}} \tau_{q}^{-2}M(\tau_q) \left[\sum_{j=n}^q\left(\begin{array}{c}q \\j\end{array}\right)\left(\begin{array}{c}q \\j-n\end{array}\right)\right]\left[q\left(\begin{array}{c}2q \\q\end{array}\right)\right]^{-1}\label{eq:sum_corr_multfrac}
\end{align}
The expectation $\tilde{E}(.)$ was taken with respect to the gaussian random variables $\xi_{t_k,\tau_l}$, excluding $M(.)$.
Repeating the procedure of Stirling's approximation and change of variable, we simplify the density,
\begin{align*}
\sum_{q=n}^{N_{\infty}}\tau_{q}^{-2}M(\tau_q) \left[\sum_{j=n}^q\left(\begin{array}{c}q \\j\end{array}\right)\left(\begin{array}{c}q \\j-n\end{array}\right)\right]\left[q\left(\begin{array}{c}2q \\q\end{array}\right)\right]^{-1}&\approx \sum_{q=n}^{N_{\infty}}n^{-2}\exp\left(-\gamma(q)^{-1}\right)\gamma(q)^{-2}M(\tau_{\gamma n^2})t_n^{-2}\\
&\approx \int_{1/n}^{N_{\infty}/n^2}d\gamma\ \exp\left(-\gamma(q)^{-1}\right)\gamma(q)^{-2}M(\gamma t_{ n})t_n^{-2}
\end{align*}
where we used $t_n \approx n \epsilon$.
Hence,
$$\tilde{E}(Y_{T}^2)=\kappa \int_0^Tdu_1\int_0^Tdu_2 \int_{0}^{\infty}d\gamma\ \exp\left(-\gamma^{-1}\right)\gamma^{-2}M(\gamma |u_1-u_2|)|u_1-u_2|^{-2}$$
for some constant $\kappa$. As higher even moments are proportial to powers of the second moment, we readily see after taking the expectation w.r.t. the distribution of $M(.)$ and using the multiplicative property (\ref{eq:multiplicative}),
$$E(Y_{T}^{2m})\propto \kappa^m \int_0^Td\vec{u} \int_{0}^{\infty}d\vec{\gamma}\ \exp\left(-\sum_j\gamma_j^{-1}\right)\left(\prod_j\gamma_j|u_{2j-1}-u_{2j}|\right)^{-2}E\left(\prod_jM(\gamma_j)\right)E\left(\prod_j M(|u_{2j-1}-u_{2j}|)\right)$$
Using the multicative property, this expression yields the sought property,
$$E\left(Y^m_{cT}\right)=E\left(M^m_{c}\right)E\left(Y^m_{T}\right) $$
from which the claim follows.
\end{proof}

\begin{align}
P\left(Y^{t_{k}}_{\tau_k},Y^{t_{k+1}}_{\tau_k}\right.&\left.|Y^{t_{k}}_{\tau_{k+1}}=z_1 ,Y^{t_{k+1}}_{\tau_{k+1}}=z_2\right)\nonumber \propto  \mathcal{N}\left(y; m_{\tau_k}^{\tau_{k+1}} (z_1+z_2), \sigma^{t_k}_{\tau_{k}},\sigma^{t_{k+1}}_{\tau_{k}} \right) 
\label{eq:MERA_transition}
\end{align}

\end{widetext}
\end{document}